\documentclass[sigconf,nonacm]{aamas} 

\usepackage{hyperref} 
\hypersetup{colorlinks=true,citecolor=blue,linkcolor=blue,urlcolor=blue}
\usepackage{amsthm}
\usepackage{amsmath}
\usepackage{amsfonts}
\usepackage{thmtools}
\usepackage[nameinlink,capitalize]{cleveref}
\usepackage{caption}
\usepackage{subcaption}
\usepackage{graphicx}
\usepackage{algpseudocode}
\usepackage{mathtools}

\DeclarePairedDelimiter\floor{\lfloor}{\rfloor}

\DeclareMathOperator{\payText}{\textit{pay}}
\newcommand{\pay}[1]{\payText({#1})}

\DeclareMathOperator{\valueText}{\textit{value}}
\newcommand{\val}[1]{\valueText({#1})}

\RequirePackage{thmtools}
\declaretheorem{falseth}
\declaretheorem[sibling=falseth]{theorem}
\declaretheorem[sibling=falseth]{lemma}

\title[Byzantine Game Theory: Sun Tzu's Boxes]{Byzantine Game Theory: Sun Tzu's Boxes}

\author{Andrei Constantinescu}
\affiliation{
  \institution{ETH Zürich}
  \city{Zürich}
  \country{Switzerland}}
\email{aconstantine@ethz.ch}

\author{Roger Wattenhofer}
\affiliation{
  \institution{ETH Zürich}
  \city{Zürich}
  \country{Switzerland}}
\email{wattenhofer@ethz.ch}

\begin{abstract}
We introduce the Byzantine Selection Problem, living at the intersection of game theory and fault-tolerant distributed computing. Here, an event organizer is presented with a group of $n$ agents, and wants to select $\ell < n$ of them to form a team. For these purposes, each agent $i$ self-reports a positive skill value $v_i$, and a team's value is the sum of its members' skill values. Ideally, the value of the team should be as large as possible, which can be easily achieved by selecting agents with the highest $\ell$ skill values. However, an unknown subset of at most $t < n$ agents are byzantine and hence not to be trusted, rendering their true skill values as $0$. In the spirit of the distributed computing literature, the identity of the byzantine agents is not random but instead chosen by an adversary aiming to minimize the value of the chosen team. Can we still select a team with good guarantees in this adversarial setting? As it turns out, deterministically, it remains optimal to select agents with the highest $\ell$ values. Yet, if $t \geq \ell$, the adversary can choose to make all selected agents byzantine, leading to a team of value zero. To provide meaningful guarantees, one hence needs to allow for randomization, in which case the expected value of the selected team needs to be maximized, assuming again that the adversary plays to minimize it. For this case, we provide linear-time randomized algorithms that maximize the expected value of the selected team. 
\end{abstract}


\begin{document}

\pagestyle{fancy}
\fancyhead{}

\maketitle 


\section{Introduction}

\textit{The Art of War} is an ancient Chinese military treatise attributed to Sun Tzu.
It is widely considered the first known origin of game theory. Sun Tzu advocates misleading the enemy through deception, discusses asymmetric information setups and even zero-sum games.

While Sun Tzu's discussions were qualitative only, around two millennia later, in the 16th and 17th centuries, the first foundations for analyzing games of chance were laid out.
John von Neumann eventually turned game theory into a rigorous academic subject,
culminating with his celebrated 1944 book \textit{Theory of Games and Economic Behavior} co-authored with Oskar Morgenstern.

Ever since, the agents participating in the game have been typically considered to be strategic, selfish, egoistic, or sometimes altruistic, but they would essentially always follow the rules of the game. In this paper, we 
take game theory back to Sun Tzu's origins of lies and deception. We believe this take on game theory opens up a whole class of exciting problems that have remained largely unexplored. 

Let us introduce our paper with a concrete puzzle.
A game show host puts four boxes in front of you, each with a value written on top: $8, 7, 5$ and $4$. Each box contains its advertised amount of value, but there is a catch: one box is a lie and contains nothing! You have to choose one box, and you will receive the money inside. 
What is your strategy to win the maximum possible amount of money? 

You 
quickly understand that there is no way to guarantee winning any money:
no matter which box you pick, it could be empty. However, 
what if you flipped a fair coin and chose the first box (with value 8) if the coin landed on heads and the second box (with value 7) if it landed on tails? At most one of those boxes can be empty, so with probability at least 50\%, you do not choose the empty box, so in expectation you earn at least $\min(8,7)/2 = 3.5$.
But can you do even better? Before reading on, we would like to encourage you to think about this puzzle for a moment.

This problem can be generalized to $n$ boxes promising monetary rewards $v_1, \dots, v_n$. At most $t < n$ boxes are a lie, and one has to choose $\ell < n$ boxes, winning the total inside these $\ell$ boxes. While we presented the problem as a game show puzzle, one can think of real-life situations where the same setup emerges: say you are running an auction for $\ell$ identical perishable items (e.g., food close to the expiration date, seat upgrades to first class on tomorrow's flight, leftover concert tickets). People place $n$
bids $v_1, \dots, v_n$, and you want to maximize your revenue. For transparency reasons, you want to publish the auction mechanism prior to the auction, and you worry that doing so might lead to the emergence of malicious bidders that bid adversarially to win the auction and then walk away after winning an item, leaving you with unsold items. Provided no binding terms enforce honoring the winning bids, your best course of action is to design your mechanism to account for the presence of a certain threshold $t$ on the number of malicious bidders.

Another example where the problem occurs is when 
we have to choose 
$\ell$ individuals out of $n$ 
applicants, like when putting together a university chess team. Again, up to $t$ applicants could be imposters, bragging about their chess strength (the number on their box) despite barely knowing the rules of chess (the box is empty).

Let us now return to our example with four boxes with values $\mathbf{v} = (8, 7, 5, 4)$. The first thought following the fair-coin strategy would be to use a biased coin to balance the expectations from the two cases: choose box 1 with probability $p_1$ and box 2 with probability $p_2 = 1 - p_1$. Then, if the first box is empty, we get an expectation of $7 \cdot p_2$, while if the second box is empty, we get an expectation of $8 \cdot p_1$. Setting the two products to be equal gives us the solution $\mathbf{p} = (\frac{7}{15}, \frac{8}{15})$, giving in both cases an expectation of $\frac{7 \cdot 8}{15} \approx 3.73$, which is better than our previous expected value of $3.5$. Going further, one would think of generalizing the idea to randomize between all four boxes: choose each box $i \in \{1, \dots, 4\}$ with probability $p_i$ such that $\sum_{i = 1}^4 p_i = 1$ and $v_1 \cdot p_1 = \dots = v_4 \cdot p_4$. Solving the ensuing equations yields the solution $\mathbf{p} = (\frac{35}{201}, \frac{40}{201}, \frac{56}{201}, \frac{70}{201})$. In all four cases (i.e., one for each possibility of which box is empty), the incurred expectation is $(4 - 1) \cdot \frac{280}{201} \approx 4.18$, which is better than our last solution --- but is this the best possible? Intuitively, one could think that the more boxes we consider, the better our expectation will be. 
However, this turns out not to be the case. In particular, ignoring box 4 is a good idea: suppose we require that $\sum_{i = 1}^3 p_i = 1$ and $v_1 \cdot p_1 = \dots = v_3 \cdot p_3$. Solving the equations yields $\mathbf{p} = (\frac{35}{131}, \frac{40}{131}, \frac{56}{131}, 0)$, incurring in all three cases (depending on which box out of the first three is empty) an expectation of $(3 - 1) \cdot \frac{280}{131} \approx 4.27$, which turns out to be the unique optimal solution.

\textbf{The Problem.} More formally, the boxes problem can be formulated as follows: $n$ boxes indexed by $[n] := \{1, \dots, n\}$ are given with promised monetary amounts $\mathbf{v} = (v_1, \dots, v_n)$ written on them. Without loss of generality, assume $v_1 \geq \dots \geq v_n > 0$.\footnote{Boxes $i$ with $v_i = 0$ can be safely ignored. This removes edge cases in the analysis.} It is known that $t < n$ of them are empty,\footnote{Alternatively, at most $t$ are empty. This does not change the analysis.} termed \emph{byzantine}, following tradition in the distributed computing literature; the other $n-t$ contain the advertised amounts. No prior over which boxes are byzantine is provided other than the fact that there are $t$ of them. The goal is to select $\ell < n$ boxes to open so as to maximize the worst-case expected total amount of money in the selected boxes. Formally, a randomized algorithm is sought that samples from a probability distribution $p$ over size-$\ell$ subsets of $[n]$ such that the quantity $\val{p}$ is maximized, which is defined as the minimum over all possibilities for the size-$t$ set $B \subseteq [n]$ of byzantine boxes of the expectation incurred by $p$ assuming boxes in $B$ are empty and the other boxes have the advertised values. See \cref{sect:prelims} for a fully formal definition and further discussion of the model.

\textbf{Our Contribution.} We give linear-time algorithms for sampling from a distribution $p$ maximizing $\val{p}$.\footnote{The underlying distribution may have exponentially-large support, so outputting it explicitly may not be possible. We also show that, at the cost of an extra factor of $n$ in the time complexity, an explicit $p$ (with linearly-sized support) can be computed.} We start by doing this for the case $\ell = 1$, for which we show that the approach we used to solve our four-boxes example is indeed correct in general: a $\val{p}$-maximizing distribution $p$ always exists such that for some $t + 1 \leq i \leq n$ we have $v_1 \cdot p_1 = \dots = v_i \cdot p_i$ and $p_{i + 1} = \dots = p_n = 0$. For a fixed $i$, there exists a single such distribution, given by setting $p_1, \dots, p_i$ proportionally to $\frac{1}{v_1}, \dots \frac{1}{v_i}$. Finding the $i$ maximizing the expectation can then be easily achieved in linear time. To show that an optimal solution maximizing $\val{p}$ has this shape, we combine an exchange argument with combinatorial reasoning about the dual of a linear program (LP). Afterwards, we move on the the general-$\ell$ case, for which such reasoning no longer suffices. Instead, we combine multiple techniques to recover a linear-time algorithm for this case as well. First, we use prior results on randomized rounding to show that it suffices to reason about the marginals of $p$ along the $n$ boxes. Then, we use a succession of exchange arguments to prove increasingly more refined forms of an optimal solution. The final form can be elegantly interpreted through a visual metaphor: pouring $\ell$ units of water into $n$ unit-volume vessels. Using this metaphor, we devise a linear-time two-pointer approach that iterates through possible solutions in a principled manner and returns an optimal one. A technical overview of our results and techniques can be found in \cref{sect:technical-overview}.

\section{Related Work}

Our paper brings together multiple fields often studied separately. In the following, we outline numerous connections between byzantine elements (and, more specifically, our selection problem) and other areas studying related notions and problems. Outlining these connections can lead to fruitful interdisciplinary future insights.

\noindent \textbf{Fault-Tolerant Distributed Computing.} The fault-tolerant distributed computing literature often considers a setup with $n$ parties, out of which at most $t$ are corrupted and may exhibit arbitrary deviations from the intended behavior. Following the tradition of the field, such parties are called \emph{byzantine}, while the other (at least) $n - t$ parties are referred to as \emph{honest} and will follow the intended behavior. One of the most prominently-studied problems is binary \emph{Byzantine Agreement} (BA), where each party $i$ has an input $x_i \in \{0, 1\}$, and must eventually provide an output $y_i \in \{0, 1\}$, such that any two honest parties give the same output (the \emph{agreement} condition) and this output is moreover the input of some honest party (the \emph{validity} condition). 
Depending on additional assumptions (e.g., synchronous communication, cryptographic setup, randomization), the difficulty of designing a correct BA protocol varies, with protocols resilient against $t < n/2$ corruptions existing in certain settings.
The core challenge here lies in the global lack of trust: nobody can be sure who are the corrupted parties, and hence protocols resilient against even a small number of corruptions can be highly non-trivial.
In contrast, assuming the existence of a \emph{trusted third party (TTP)} that is guaranteed to be honest makes BA essentially trivial: all parties send their $x_i$'s to the TTP, which takes a majority vote and then sends back the outcome to all parties. Consequently, efforts in the literature boil down to simulating the behavior of the TTP in an untrusted setting. The keen reader might have already noted that the simple TTP protocol above fails to work past the $t < n/2$ corruption threshold (if too many parties are byzantine, they can outnumber the honest votes), and, in fact, no TTP protocol can surpass this bound, implying that $t < n/2$ is best possible in the distributed setting too.
Hence, from a fresh point of view, understanding the difficulty of problems in fault tolerance begins with understanding their difficulty in the TTP setting. While this was mostly straightforward for binary BA (and is hence rarely considered), the story becomes interesting again for agreement problems over other domains:

When the inputs are real numbers, \emph{Honest-Range Validity} \cite{aa_with_range_validity} requires that the agreed-upon value is between the smallest and largest honest inputs. A stronger requirement stands in variations of \emph{Median Validity} \cite{stolz,clocks,kth_element}, asking that the output is close to the median of the honest inputs. An orthogonal generalization of \emph{Honest-Range Validity} is when the inputs are $D$-dimensional real vectors, where \emph{Convex Validity} \cite{vaidya_garg_convex,convex_aa} requires the output vector to be in the convex hull of the honest input vectors. This has been further generalized to abstract convexity spaces, including graphs and lattices \cite{generalized_convexity_space_us,nowak_rybicki}. When the inputs are linear orders over a set of alternatives, also known as \emph{rankings} in social-choice theoretic terms, \cite{darya-byzantine-voting} considers \emph{Pareto Validity}, requiring that whenever all honest inputs rank $a$ above $b$, so does the output ranking. The same paper also studies a second problem, requiring an output ranking that is close to the \emph{Kemeny median} of the honest rankings. Another notion inspired by social-choice is that of \emph{Voting Validity} \cite{voting-validity}.

Note how all these problems are non-trivial even in the TTP setting: how should the TTP aggregate $n$ votes to yield a valid outcome when $t$ of the votes might be corrupted and should not be considered? The problem we study is precisely of this flavor: the TTP receives the box values $v_1, \dots, v_n$ and must output a set of boxes to open to minimize the loss from empty boxes\footnote{Technically, maximize the win from non-empty boxes (the sum of the two is non-constant, so there is a distinction).} (possibly using randomization). In general, recent results indicate that the difficulty of solving the ``centralized'' TTP problem often matches that of its decentralized version, at least when considering BA with non-trivial validity conditions \cite{pierre_civit_1,pierre_civit_2,constantinescu2024validitynetworkagnosticbyzantineagreement}. Hence, it is imperative to understand the difficulty of the centralized problem in more settings, and our paper sets out to do just this for our selection problem.

\noindent\textbf{Malice in Game Theory.} While most of game theory considers rational (selfish) agents, a few works consider the presence of malicious agents, who derive their utility from making the system perform poorly or from others' disutility. This has been particularly investigated in distributed settings, where considering malicious actors is the norm. An excellent survey of work at the interface of rational and malicious behavior can be found in \cite{the_price_of_malice}, which also introduces the \emph{Price of Malice}, quantifying the degradation of the performance of a system of rational agents with the introduction of malicious actors.\footnote{The notion is inspired by the celebrated \emph{Price of Anarchy}, which instead targets the difference between selfish and collaborating agents.} The same paper analyzes the price of malice in a virus inoculation game, and subsequent papers study it and other notions of malice in other classes of games \cite{roth_congestion, babaioff, chakrabarty2009effectmalicesocialoptimum}. Particularly relevant to the distributed computing field is the concept of \emph{BAR fault tolerance} \cite{aiyer}, which requires distributed protocols to withstand \textbf{B}yzantine, \textbf{A}ltruistic, and \textbf{R}ational (BAR) behavior. Specifically, they must tolerate a constant fraction of Byzantine agents, as is standard, while ensuring that the remaining agents --- who are rational --- have sufficient incentives to follow the protocol. Subsequent work adopted the notion and applied it in a variety of distributed settings, spanning theory and practice \cite{bar_2, bar_3, FlightPath}. A number of works define fault-tolerant solution concepts, e.g., fault-tolerant Nash Equilibria, and apply them to study the fault tolerance of various games \cite{eliaz, karakostas_viglas}. Also relevant are the appealing results in \cite{gradwohl}, which find that ``large'' games are naturally fault-tolerant. Malice can also manifest as \emph{spite}. In this context, \cite{brandt, MorganSteiglitzReis} examine auctions with spiteful agents, i.e., agents who derive utility from others' disutility. Conversely, several works consider altruistic agents (which, in the previous context, referred to agents who follow the protocol regardless of incentives, but this is not a strict requirement), modeled as agents for whom larger utilities of others translate to larger utilities for themselves \cite{hoefer_skopalik, meier_yvonne_stefan_roger}.

\noindent \textbf{Bribery in Voting.} The social choice literature considers election \emph{bribery} \cite{Faliszewski_Rothe_2016}: changing the outcome 
by bribing a subset of voters to change their ballots. The aggregation problem that the TTP faces in distributed computing is of a similar flavour: the byzantine votes can be thought of as bought votes, and the bound on the number of corruptions $t$ as the budget of the briber. However, there is a notable difference: in voting the aggregation mechanism is fixed and the aim is to quantify the damage that can be done for a given profile of votes given a budget $t$. In contrast, our goal is designing a robust aggregation mechanism that is aware of the presence of at most $t$ ill-intended votes (which should be rightfully nullified if ever identified, although no signals to this end exist in our setting).

\noindent \textbf{Stackelberg and Security Games.} Our problem can be seen as a \emph{zero-sum Stackelberg (maxi-min) game} where the leader commits to a (potentially randomized) strategy of picking $\ell$ boxes, and then the follower (i.e., the adversary), knowing the leader's choice, chooses $t$ boxes to nullify. The leader (follower) attempts to maximize (minimize) the expected sum of the $\ell$ selected boxes post-nullification.
Related is the class of \emph{security games}, which have received extensive attention in recent years \cite{security_survey}. In one variant \cite{security_multiple_attackers}, there are $n$ potential targets, an attacker, and a defender. The defender moves first, chooses $t$ targets to protect (using a potentially randomized strategy), and then the attacker chooses $\ell$ targets to attack, knowing the strategy of the defender. The utilities of the two players additively depend solely on attacked targets, with the goal of the attacker being to attack undefended targets and that of the defender being to protect attacked
targets, but the game is not necessarily zero-sum. The existing literature largely concerns the non-zero-sum case with $\ell = 1$ attacked targets. On the computational front, security games are largely amenable to techniques from combinatorial optimization \cite{security_versus_combinatorial}, often linear programming \cite{conitzer_lp_security}. Results become scarcer when seeking more efficient algorithms: \cite{origami} for the $\ell = 1$ case, and \cite{security_multiple_attackers} for the general-$\ell$ case if instead of Stackelberg Equilibria we require Nash Equilibria, the former becoming computationally demanding. Our problem corresponds to a zero-sum variant of the previously described security game with the order the attacker and defender play in reversed. We note that this does not fundamentally change the game since for zero-sum games, the maxi-min and mini-max values coincide by von Neumann's theorem \cite{Neumann1928}. Stackleberg and Nash Equilibria are also closely tied in the zero-sum case, so the algorithm of \cite{security_multiple_attackers} can be used for our setting, but it is arguably more complicated than our approach (and quadratic instead of linear).

\noindent \textbf{Robust Combinatorial Optimization.} An active area of operations research concerns optimization under uncertainty~\cite{bertsimas2022robust}: when there is uncertainty in the constraints or objective. 
For our purposes, let us restrict ourselves to a problem template where only the objective is uncertain: the (not necessarily continuous) feasible region is known, denoted by $\mathcal{F} \subseteq \mathbb{R}^n$, and we are interested in $\max_{x \in \mathcal{F}} c^T x$. Instead of knowing $c$, we only know an uncertainty set $\mathcal{U}$ such that $c \in \mathcal{U}$. No probability distribution over $\mathcal{U}$ is supplied: we seek a solution $x$ that maximizes $c^T x$ in the worst-case, making for the \emph{maxi-min objective} $\max_{x \in \mathcal{F}} \min_{x \in \mathcal{U}} c^T x$. This setup is very flexible, as $\mathcal{F}$ can range from a polytope in the continuous case to the set of $s$-$t$ paths or spanning trees of a graph in the discrete one. The maxi-min objective can be replaced with regret-inspired variants. The uncertainty set $\mathcal{U}$ can take various shapes, with the two most prominent ones being \emph{discrete uncertainty}: $\mathcal{U}_D = \{c_1, \dots, c_k\}$ and \emph{interval uncertainty}: $\mathcal{U}_I = [a_1, b_1] \times \dots \times [a_n, b_n]$. Interval uncertainty admits a variant in the spirit of byzantine fault-tolerance, introduced in \cite{threshold_adversary_is_defined}: given a threshold $\Gamma$ define \emph{$\Gamma$-interval uncertainty} $\mathcal{U}_I^\Gamma$ such that $c \in \mathcal{U}_I^\Gamma$ if $c \in \mathcal{U}_I$ and $|\{i : c_i \neq b_i\}| \leq \Gamma$. See \cite{robust_survey1, robust_survey2, robust_survey3_algorithmic_design, robust_survey4, robust_survey5_practical} for excellent surveys of results and techniques. One of the basic cases considered in the robust optimization literature concerns $\mathcal{F} = \{x \in \{0, 1\}^n : \sum_{i = 1}^n x_i = \ell\}$, the so-called \emph{Selection Problem}: see \cite{robust_survey2} for a compilation of results under various objective and uncertainty set assumptions. Of particular interest to us is the case with the normal maxi-min objective and $\Gamma=t$-interval uncertainty for $\mathcal{U}_I = [0, v_1] \times \dots \times [0, v_n]$: this corresponds exactly with choosing to open $\ell$ boxes, out of which the adversary can nullify $t$. Note, however, that this only models the deterministic part of our paper (which is straightforward): it cannot model committing to a randomized strategy of which $\ell$ boxes to open to which the adversary replies by nullifying $t$ so as to minimize the expectation. One of the few papers considering such randomized strategies is \cite{selection_minmax_regret_randomized}, where the authors show that for discrete and interval uncertainty sets, it is possible to optimize the regret objective in polynomial time as long as (non-robust) optimization over $\mathcal{F}$ is feasible polynomially (which trivially holds for the Selection Problem). Note that their result does not target $\Gamma$-interval uncertainty and is for the regret objective, hence not applicable to us. 
The paper \cite{selection_minmax_regret_randomized} also refers to an unpublished (and not publicly available) 2012 manuscript of Bertsimas, Nasrabadi, and Orlin, titled ``On the power of nature in robust discrete optimization,'' claiming a similar result for our maxi-min objective, this time more generally applying to any pair $(\mathcal{F}, \mathcal{U})$ such that (non-robust) optimization over both $\mathcal{F}$ and $\mathcal{U}$ is feasible polynomially (which is the case in our setting). Modulo the fact that their paper cannot be reasonably retrieved, their result implies polynomial solvability for our problem (most likely using continuous optimization techniques like LP, and hence not strongly-polynomial). Our approach to the problem will be different, leading to a better, linear-time, algorithm.

\noindent \textbf{Statistical Learning with Adversarial Noise.} Learning the underlying distribution or statistics about a dataset is a fundamental problem in statistical learning. However, real-world data is often ill-behaved, including a fraction of adversarially corrupted/byzantine data. An extensive line of work has been dedicated to robust learning in this setting \cite{robust_mean_covariance_1, robust_geometric_concept_classes, robust_gaussians_1, robust_gaussians_2, learning_from_untrusted_data}, including but not limited to learning high-dimensional Gaussians with adversarial noise.

\noindent \textbf{Adversarial Bandits and Experts.} \emph{Bandit} and \emph{Expert Learning} are two related cornerstone models of decision-making under uncertainty, with applications ranging from machine learning to theoretical computer science and economics \cite{bandits1, bandits2, bandits3}. In the former, a decision-maker faces a slot machine with $n$ arms with unknown, possibly different, reward distributions. At each time-step, having observed past behavior, the decision-maker must select an arm to pull, gaining an instantaneous reward sampled from that arm's distribution independently of other time-steps.
The standard goal is regret minimization across a (possibly not fixed) time horizon $T$. Adversarial variants of the model have been explored, where instead of independent stochastic rewards, the rewards in each time slot are chosen by an adversary, bringing the setup closer to ours. Our problem can be seen as a single-shot ($T = 1$) variant of adversarial bandits. Here, each arm gives a known reward of either $0$ or $v_i$, and the number of zeros is limited by $t$ (the number of Byzantine arms), but it is not known which $t$ arms give zeros. The goal also changes from regret minimization to value maximization. Some work also explored models in-between stochastic and adversarial bandits \cite{stochastic_adversarial_bandits,stochastic_adversarial_bandits_2}, closer to our byzantine setting.
More broadly, there is interest in such in-between models \cite{stochastic_adversarial}.

\noindent \textbf{Secretary Problem.} In the classical \emph{Secretary Problem} \cite{who_solved_secretary, dynkin}, a stream of $n$ items arrive one by one in an online fashion. After each arrival, the algorithm should either commit to that item or discard it permanently. Each item has a value, and common goals include maximizing the probability of picking a highest-value item or maximizing the expected value of the picked item. A plethora of variations of the problem have been explored, including picking $\ell$ items instead of a single one \cite{k_secretary}, connecting the problem to our paper. In the standard secretary model, the adversary picks the values of the items knowing the algorithm, and then the items are presented to the algorithm in an order chosen uniformly at random (the \emph{random-order} model). Most developed algorithms are not robust to even small adversarial perturbations of the random order assumption. To study this effect, \cite{goran} introduces a semi-random model, where all values stay adversarial, but the arrival times of $t$ items are adversarial, chosen before the other $n - t$ items' arrival times are randomly generated. Like us, inspired by the distributed computing literature, they call this the \emph{Byzantine Secretary} model. We note that the problem in \cite{goran} is very different from ours: they consider an online setting where all values are adversarial, but the arrival times of only $t$ are adversarial, while our setting is offline, with all but $t$ of the values being honest. More subtly, given the values $v_1, \dots, v_n$, we want 
the best possible solution for a given instance, instead of the best achievable for a worst-case instance as in their setting (making our problem harder in this regard).

\noindent \textbf{Prophet Inequalities.} A commonly studied online stopping problem, related to the Secretary Problem, but also to Bandit Learning due to its online flavor, goes under the name of \emph{Prophet Inequalities} \cite{prophet}. The basic setup considers $n$ random variables with known distributions $X_1, \dots, X_n$. One by one, the realization of these random variables is revealed to the decision-maker, who has the option between committing to the current value or passing to the next round. The goal is to maximize the expectation of the selected value, and the standard result is that a policy achieving an expectation of at least half of the ex-post maximum value exists. Variants where $\ell$ values are to be selected have been studied, but once again, the combination of the stochastic and online aspects makes this setup very different from ours.

\section{Preliminaries}\label{sect:prelims}

\textbf{Sets and Distributions.} Given a non-negative integer $k,$ write $[k] := \{1, 2, \dots, a\}.$
Given a set $S$ and a non-negative integer $k$, we write $S^{(k)} := \{S' \subseteq S \mid |S'| = k\}$ for the set of $k$-element subsets of $S$. 
For finite $S$, we write $\Delta(S) := \{p : S \to [0, 1] \mid \sum_{s \in S}p_s = 1\}$ for the set of probability distributions over $S$.
For technical reasons, given a non-negative number $x$, we also define $\Delta_x(S) := \{p : S \to [0, 1] \mid \sum_{s \in S}p_s = x\},$ the set of what we will call \emph{pseudo-distributions} of sum $x$ over $S$. Note that $\Delta_1(S) = \Delta(S).$

\textbf{Setup.} We consider a setting with $n$ boxes indexed by the set $[n].$ Each box $i \in [n]$ has a real number $v_i > 0$ written on it.
For ease of presentation, except where stated otherwise, we will assume that $v_1 \geq \dots \geq v_n$. Each box $i$ may either be \emph{honest}, in which case it contains $v_i$ units of money, or \emph{byzantine}, in which case it contains no money, without loss of generality. Assume $B \subseteq [n]$ is the set of byzantine boxes, then $[n] \setminus B$ is the set of \emph{honest} boxes. Given the set of byzantine boxes $B$, the payoff incurred by opening box a box $i$ is 
$\pay{i, B} = v_i$ if $i \notin B$ and $\pay{i, B} = 0$ otherwise. The total payoff incurred by opening a set of boxes $S \subseteq [n]$ is defined additively as $\pay{S, B} := \sum_{i \in S}\pay{i, B}.$ 

\textbf{Problem Definition.} We are interested in designing a (for now deterministic) mechanism $M$ that, given $n$, the mapping $v : [n] \to \mathbb{R}^+$ 
and two numbers $1 \leq t, \ell < n$ select a size-$\ell$ subset of boxes $M(n, v, t, \ell) = S \subseteq [n]$ to open such that the worst-case payoff is maximized with respect to all options for the size-$t$ subset $B \subseteq [n]$ of byzantine boxes.\footnote{Note that requiring $|S| \leq \ell$ and $|B| \leq t$ instead of $|S| = \ell$ and $|B| = t$ would not change the problem.} 
Formally, we want to design a $M$ such that:
\begin{equation}\label{eq:goal-deterministic}
    M(n, v, t, \ell) \in \operatorname*{argmax}_{S \in [n]^{(\ell)}} \min_{B \in [n]^{(t)}}\pay{S, B}
\end{equation}

The right-hand side of \cref{eq:goal-deterministic} can be seen as a game between the mechanism and an adversary: the mechanism picks the set $S$ of boxes to open, then the adversary (knowing $S$) selects the set $B$ of byzantine boxes. This is a zero-sum game where the mechanism aims to maximize the payoff, while the adversary to minimize it.\footnote{In fact, \cref{eq:goal-deterministic} corresponds to the game-theoretic maxi-min solution concept.} 

Let us define $\val{S}:= \min_{B \in [n]^{(t)}}\pay{S, B}$, in which case \cref{eq:goal-deterministic} asks that the selected $S$ maximizes $\val{S}$. Our problem, as stated so far, is easy to solve:~from any set $S$ of selected boxes, the adversary will choose to nullify the largest $\min\{t, \ell\}$ values by including the corresponding boxes in $B.$ Hence, to maximize payoff, the mechanism should select boxes with the $\ell$ largest values, from which the largest $\min\{t, \ell\}$ will be nullified by the adversary. 

\begin{theorem} Among \textbf{deterministic} mechanisms, selecting $S^* = \{1, 2, \dots, \ell\}$ achieves the highest possible payoff:%
\begin{equation*}
    \max_{S \in [n]^{(\ell)}} \min_{B \in [n]^{(t)}}\pay{S, B} = 
    \val{S^*} = \left\{
\begin{array}{ll}
      \sum_{i = t + 1}^{\ell}v_i & t < \ell\\
      0 & t \geq \ell\\
\end{array} 
\right.
\end{equation*}
\end{theorem}

As an immediate consequence, deterministic mechanisms fail to provide any meaningful guarantees in the simplest case $t = \ell = 1$.  Naturally, the next step is to allow for randomization, in which case we need to discuss the power of the adversary, i.e., how much they are allowed to know about the random decisions of the mechanism prior to selecting $B$. 
For instance, a \emph{strong adversary}\footnote{Also known as an \emph{adaptive offline adversary} in the context of online algorithms.} knows both the mechanism and the randomness before picking which boxes are byzantine.
Since knowing the randomness renders any randomized mechanism deterministic, randomization does not help against a strong adversary. On the other hand, an \emph{oblivious/weak adversary} has access to the mechanism, but not to the random bits used. Formally, the game played against an oblivious adversary proceeds as follows: the mechanism outputs a probability distribution over sets $S$ of $\ell$ boxes each, say $p \in \Delta\left([n]^{(\ell)}\right)$, then the adversary (knowing $p$) selects the set $B$ of at most $t$ byzantine boxes,\footnote{Since the adversary plays second, it does not help them to randomize their strategy.} and finally a set $S \sim p$ is sampled, incurring a payoff of $\pay{S, B}.$ Before the final set $S$ is sampled, the outcome of the game is a probability distribution over possible payoffs, allowing for the formulation of various optimization goals. Most prominently, one can optimize for the expected payoff, but this is not the only option: a risk-averse mechanism user might prefer an expected payoff of $99$ that guarantees a payoff of at least $50$ in all realizations of the randomness to an expected payoff of $100$ that leaves a positive probability of getting payoff $1$.
In the paper, we will assume the expectation objective and an oblivious adversary, unless stated otherwise. With these assumptions
we seek a randomized mechanism $M$ such that:%
\begin{equation}\label{eq:goal-randomized}
    M(n, v, t, \ell) \in \operatorname*{argmax}_{p \in \Delta\left([n]^{(\ell)}\right)} \min_{B \in [n]^{(t)}}\mathbb{E}_{S \sim p}[\pay{S, B}]
\end{equation}

As before, we define $\val{p} := \min_{B \in [n]^{(t)}}\mathbb{E}_{S \sim p}[\pay{S, B}]$, in which case \cref{eq:goal-randomized} asks that the selected $p$ maximizes $\val{p}$.

\section{Results and Technical Overview}\label{sect:technical-overview}
The case $\ell = 1$ allows for simpler arguments and characterizations of optimal solutions (and generally less machinery). To solve it, we will use a natural exchange argument to show that an optimal solution satisfying $v_1 \cdot p_1 \geq \dots \geq v_n \cdot p_n$ always exists. Restricting our search to such distributions $p$ is particularly useful because we then know that the worst-case set of byzantine boxes is $\{1, \dots, t\}.$
These observations let us cast the problem as a simple linear program (LP). 
Reasoning combinatorially about its dual then shows that an optimal solution taking a very elegant form always exists. In particular, this solution satisfies $v_1\cdot p_1 = \dots 
= v_i\cdot p_i$ and $p_{i + 1} = \dots = p_n = 0,$ for some $t + 1 \leq i \leq n.$ In other words, it is optimal to select a prefix of boxes $1, \dots, i$ of length at least $t + 1$ and put all the probability mass on it such that the expected values $v_j \cdot p_j$ are equal on the prefix. For a fixed $i$, the resulting distribution is unique and given by setting $\left(p_1, \dots, p_n\right)$ proportionally to $\left(\frac{1}{v_1}, \dots, \frac{1}{v_i}, 0, \dots, 0\right)$, in which case one can check that $\val{p} = \frac{i - t}{\sum_{j = 1}^i \frac{1}{v_j}}$. It is then easy to compute this value in linear time for all $t + 1 \leq i \leq n$ and take the maximum. We prove the required claims in section \cref{sect:ell-1}.

\begin{theorem}\label{th:main:ell-one} Assume $\ell = 1$ and define $p^i$ for $t + 1 \leq i \leq n$ to be the unique distribution such that $\left(p^i_1, \dots, p^i_n\right)$ is proportional to $\left(\frac{1}{v_1}, \dots, \frac{1}{v_i}, 0, \dots, 0\right)$, in which case $\val{p^i} = \frac{i - t}{\sum_{j = 1}^i \frac{1}{v_j}}$. Then, among distributions $p \in \Delta([n])$, the maximum $\val{p}$ is attained at one of $p^{t + 1}, \dots, p^n$, and we can determine which one in linear time.
\end{theorem}

For the harder general case, it is no longer the case that a value-maximizing distribution of such an attractive shape exists. First and foremost, this is because we are now dealing with distributions over size-$\ell$ subsets of $[n]$. To overcome this first obstacle, linearity of expectation gives us that $\val{p}$ for some $p \in \Delta([n]^{(\ell)})$ is uniquely determined by the \emph{marginals} $p'_i := \mathbb{P}_{S \sim p'}(i \in S)$ for $i \in [n]$. Naturally, the marginals are between $0$ and $1$ and sum up to $\ell$ since the sampled set $S$ satisfies $|S| = \ell$, from which $p' \in \Delta_\ell([n])$. Note $p'$ is \emph{not} just a distribution scaled up by a factor of $\ell$ because of the constraint that $p'_i \in [0, 1]$ for $i \in [n]$. Still, it would be convenient to optimize directly over $p' \in \Delta_\ell([n])$, a much lower-dimensional object. Perhaps surprisingly, this is something that we will be able to do by invoking results on randomized rounding: for any $p' \in \Delta_\ell([n])$, there exists $p \in \Delta([n]^{(\ell)})$ such that $p'$ gives the marginals of $p$, i.e., $p'_i = \mathbb{P}_{S \sim p}(i \in S)$ for $i \in [n]$, and we can sample from such a $p$ in linear time.
Hence, our new goal will be to find a $p' \in \Delta_\ell([n])$ maximizing the appropriately redefined $\val{p'}$.

Having retaken the problem into the realm of combinatorial tractability, we then use a similar exchange argument to show that a value-maximizing $p'$ exists such that $v_1 \cdot p'_1 \geq \dots \geq v_n \cdot p'_n$. Past this point,  unfortunately, optimal solutions no longer exhibit the elegant form that we could prove by considering the dual in the $\ell = 1$ case.\footnote{We can still cast the problem as an LP and take its dual, but this requires adding the constraints $p'_i \leq 1$ to the primal since they no longer follow from the sum constraint, leading to a dual that is harder to analyze and no longer implies the required property.} To pass this hurdle, we take a different approach, proving a series of more refined forms of the previous inequality. The last of these can be attractively visualized through a physical metaphor involving pouring $\ell$ units of liquid into a histogram consisting of $n$ rectangular vessels where box $i$ corresponds to a $v_i \times (1 / v_i)$ vessel. The amount of water poured into each vessel $i$ corresponds to $p'_i,$ justifying why the rectangles were chosen to have area 1. The pseudo-distributions $p'$ that we will need to consider will be uniquely determined by the water level in the first container. This way, we will simulate in linear time (in a two-pointer fashion) the process of continuously decreasing the water level in the first container and find the moment in time where the value of the corresponding $p'$ is maximized. We note that this algorithm can, of course, also be used when $\ell = 1$, but the details are considerably trickier and do not fully exploit the structure present in that case. All these considerations are carried out in \cref{sect:ell-general}.

\begin{theorem}\label{th:main} Assume $\ell \geq 1$. We can sample from a distribution $p \in \Delta\left([n]^{(\ell)}\right)$ that maximizes $\val{p}$ in linear time. If required, an explicit such $p$ can be computed in $O(n^2)$ time. 
\end{theorem}

\section{The Case $\ell = 1$}\label{sect:ell-1}

As a warm-up, to build intuition and avoid some notational burdens of the general case, we begin with the case $\ell = 1,$ where the goal is to select a single ``winner'' box with the highest value possible. In the presence of $t$ byzantine boxes, for randomized mechanisms, this means outputting a distribution $p \in \Delta([n])$ maximizing
\begin{equation}\label{eq:objective:ell1}
    \val{p} = \min_{B \in [n]^{(t)}}\mathbb{E}_{i \sim p}[\pay{i, B}] = \min_{B \in [n]^{(t)}} \sum_{i \notin B} v_i \cdot p_i = \sum_{i = t + 1}^{n} x_i 
\end{equation}

\noindent where $x_1, \dots, x_n$ are the values $v_1 \cdot p_1, \dots, v_n \cdot p_n$ ordered such that $x_1 \geq \dots \geq x_n$. The first equality follows by linearity of expectation, while the second holds because the best strategy for the adversary is to choose the byzantine boxes have the highest $t$ values $v_i \cdot p_i.$

\begin{lemma}\label{lemma:increasing:ell1} There exists a distribution $p \in \Delta([n])$ maximizing $\val{p}$ that satisfies $v_1 \cdot p_1 \geq \dots \geq v_n \cdot p_n.$
\end{lemma}
\begin{proof} Among distributions $p$ maximizing $\val{p}$, consider one such that for any $i < j$ with $v_i = v_j$ we have $p_i \geq p_j$. We show that this $p$ has the required property. Assume the contrary, then for some $i < j$ we have $v_i \cdot p_i < v_j \cdot p_j$. Recall that $v_i \geq v_j$, and moreover, since in case $v_i = v_j$ the condition $v_i \cdot p_i < v_j \cdot p_j$ reduces to $p_i < p_j,$ which can not hold by our choice of $p,$ we have $v_i > v_j.$ Let us now construct $p' : [n] \to [0, 1]$ as follows:
\begin{equation*}
    p'_k = \left\{
\begin{array}{ll}
      \frac{v_j \cdot p_j}{v_i} & k = i \\
      \frac{v_i \cdot p_i}{v_j} & k = j \\
      p_k                  & k \notin \{i, j\}
\end{array} 
\right.
\end{equation*}
First, note that the values $v_1\cdot p_1, \dots, v_n \cdot p_n$ and $v_1\cdot p'_1, \dots, v_n \cdot p'_n$ coincide except for the $i$'th and $j$'th entries being swapped, so $\val{p} = \val{p'}$.
Furthermore, it turns out that $p'$ does not use the whole available probability mass:
\begin{gather*}
    \sum_{k = 1}^{n}p'_k = p'_i + p'_j + \sum_{\substack{k = 1\\ k \notin\{i, j\}}}^{n}p_k = p'_i + p'_j + 1 - p_i - p_j = \\
    \frac{v_j \cdot p_j}{v_i} + \frac{v_i \cdot p_i}{v_j} + 1 - p_i - p_j = 
    1 - (v_i - v_j)\left(\frac{v_j \cdot p_j - v_i \cdot p_i}{v_i \cdot v_j}\right) < 1
\end{gather*}
Where the last inequality holds because $v_i > v_j$ and $v_i \cdot p_i < v_j \cdot p_j$. Let $\alpha > 1$ be such that $\sum_{k = 1}^{n}p'_k = \frac{1}{\alpha}.$ Then, the distribution $p'' \in \Delta([n])$ given by $p''_k = \alpha \cdot p'_k$ satisfies $\val{p''} = \alpha \cdot \val{p} > \val{p}$, contradicting the optimality of $p$.
\end{proof}

With this lemma in place, let us restrict ourselves to distributions $p$ satisfying $v_1 \cdot p_1 \geq \dots \geq v_n \cdot p_n.$ For this case \cref{eq:objective:ell1} simply becomes $\val{p} = \sum_{i = t + 1}^{n} v_i \cdot p_i.$ With this observation, we now prove the following lemma, implying \cref{th:main:ell-one}:

\begin{lemma}\label{lemma:exact-shape-ell-one} There exists a distribution $p \in \Delta\left([n]\right)$ maximizing $\val{p}$ such that for some $i \geq t + 1$ we have $v_1 \cdot p_1 = \dots = v_i \cdot p_i$ and $p_{i + 1} = \dots = p_n = 0$. 
\end{lemma}

\begin{proof} By \cref{lemma:increasing:ell1}, requiring that $v_1 \cdot p_1 \geq \dots \geq v_n \cdot p_n$ does not change the maximum achievable $\val{p}$, so let us restrict ourselves to this case. As a result, we know that $\val{p} = \sum_{i = t + 1}^n v_i \cdot p_i$, so a maximizing $p$ is a solution to the following linear program with $n$ variables and $2n$ constraints:%
\begin{equation*}
\begin{array}{lll}
\text{maximize}  & \sum_{i = t + 1}^n v_i \cdot p_i &\\
\text{subject to}& \ v_{i + 1} \cdot p_{i + 1} - v_i \cdot p_i \leq 0, & i = 1, \dots, n - 1 \\ 
                 & \ p_i \geq 0, & i = 1, \dots, n \\
                 & \ \displaystyle\sum_{i = 1}^{n} p_i = 1
\end{array}
\end{equation*}

Let us take its dual. For this purpose, introduce non-negative variables $y_1, \dots, y_{n - 1}$ corresponding to constraints of the first kind and a variable $T \in \mathbb{R}$ corresponding to the last constraint. For convenience, we define $y_0 = y_n = 0$ and $v'_1, \dots, v'_n$ such that $v'_i = v_i$ if $i \geq t + 1$ and $v'_i = 0$ otherwise. Then, the dual is:%
\begin{equation*}
\begin{array}{ll@{}ll}
\text{minimize}  & T & & \\
\text{subject to}& & T - v_i \cdot y_i + v_i \cdot y_{i - 1} \geq v'_i, & i = 1, \dots, n \\ 
                 & & y_i \geq 0, & i = 1, \dots, n - 1 \\
                 & & y_i = 0     & i \in \{0, n\}
\end{array}
\end{equation*}

We want to find the optimum value of the dual. This amounts to understanding for which values of $T$ the dual is feasible. Let us fix a value $T$ and study the feasibility of the dual. 

Note that constraints of the first kind in the dual can be conveniently rewritten as:%
\[
    y_i - y_{i - 1} \leq \frac{T - v'_i}{v_i} = \frac{T}{v_i} - [i \geq t + 1]
\]

\noindent where $[i \geq t + 1] := 1$ if $i \geq t + 1$ and $0$ otherwise. Denote the right-hand side of the inequality with $c^T_i$, which is a constant in terms of the fixed $T$. Hence, we want to check the feasibility of:\footnote{The accustomed reader will notice the similarity with the linear program for distances in the graph $n \rightarrow \dots \rightarrow 0$.}
\begin{equation*}
\begin{array}{ll}
y_i - y_{i - 1} \leq c^T_i, & i = 1, \dots, n \\ 
y_i \geq 0, & i = 1, \dots, n - 1 \\
y_i = 0     & i \in \{0, n\}
\end{array}
\end{equation*}

It is natural to reinterpret this in terms of $\delta_i := y_i - y_{i - 1}$:
\begin{equation*}
\begin{array}{ll}
\delta_i \leq c^T_i, & i = 1, \dots, n \\ 
\sum_{j = 1}^i \delta_j \geq 0, & i = 1, \dots, n - 1 \\
\sum_{j = 1}^n \delta_j = 0     &
\end{array}
\end{equation*}

If we relax the last constraint to $\sum_{j = 1}^n \delta_j \geq 0$, this is clearly feasible if and only if setting $\delta_i = c^T_i$ satisfies the non-negativity constraints. Doing so might however lead to $\sum_{j = 1}^n \delta_j > 0$. If this is the case, it suffices to decrease $\delta_n$ to make the total equal 0, so this is not a problem.

As a result, a value of $T$ is achievable for the dual if and only if $\sum_{j = 1}^i c^T_j \geq 0$ for all $1 \leq i \leq n$. For a fixed $i$, this is equivalent to:
\begin{gather*}
   \sum_{j = 1}^i \frac{T}{v_j} - \sum_{j = 1}^i [j \geq t + 1] = T \cdot \sum_{j = 1}^i \frac{1}{v_j} - \max\{0, i - t\} \geq 0 \\ \iff 
   T \geq \frac{\max\{0, i - t\}}{\sum_{j = 1}^i \frac{1}{v_j}}
\end{gather*}

For $i \leq t$, this simplifies to $T \geq 0$, and otherwise it simplifies to a stricter inequality, so the case $i \leq t$ can be discarded from the condition. We hence get that the optimum value of the dual is:
\[
T^* = \max_{t + 1 \leq i \leq n}\frac{i - t}{\sum_{j = 1}^i \frac{1}{v_j}}
\]

We will now give a feasible solution $p^*$ to the primal with the property required in the statement of the lemma. We will show that $p^*$ achieves value $T^*$ implying, by weak duality, that it is optimal for the primal, concluding the proof.

The distribution $p^*$ is constructed as follows: for $j > i^*$, we set $p^*_j = 0$, while for $j \leq i^*$ we set $p^*_j = \frac{1}{v_j \cdot C}$, where $C := \sum_{k = 1}^{i^*}\frac{1}{v_k}$ is a normalizing factor. First, $p^*$ is clearly a well-defined distribution, as its entries sum up to 1 and are non-negative. Moreover, $v_j \cdot p_j$ is $\frac{1}{C}$ for $j \leq i^*$ and $0$ for $j > i^*$, implying that it is a feasible solution for the primal with the property required in the lemma statement. It remains to show that $p^*$ achieves a value of $T^*$ in the primal to get the optimality of $p^*$, completing the proof. This amounts to a simple algebraic verification.
\end{proof}

\section{The General-$\ell$ Case}\label{sect:ell-general}

The general case concerns value-maximizing distributions $p \in \Delta([n]^\ell)$, i.e., over size-$\ell$ subsets of $[n]$. 
For $\ell > 1$, these are combinatorially difficult to reason about. However, by linearity of expectation, we can write $\mathbb{E}_{S \sim p}[\pay{S, B}] = \sum_{i \notin B}v_i \cdot \mathbb{P}_{S \sim p}(i \in S),$ so much of the information contained in $p$ is redundant:~the expectation only depends on the \emph{marginals} $p_i' := \mathbb{P}_{S \sim p}(i \in S)$ for all $i \in [n].$ In particular, $\mathbb{E}_{S \sim p}[\pay{S, B}] = \sum_{i = 1}^{n}v_i \cdot p'_i$. Note that the marginals define a pseudo-distribution $p' \in \Delta_\ell([n])$: they are between $0$ and $1$ and sum up to $\ell$.\footnote{By linearity of expectation: $\sum_{i = 1}^{n}p'_i = \sum_{i = 1}^{n}\mathbb{P}_{S \sim p}(i \in S) = 
\mathbb{E}_{S \sim p}[|S|] = \ell$.} Trying to optimize for $p'$ directly seems particularly attractive given its lower-dimensional nature. In particular, instead of optimizing for $\val{p}$ among $p \in \Delta([n]^\ell)$, we can optimize for $\val{p'} := \min_{B \in [n]^{(t)}}\sum_{i = 1}^{n}v_i \cdot p'_i$ among $p' \in \Delta_\ell([n])$. While attractive, this approach could be prone to a serious pitfall: it could be that value-maximizing pseudo-distributions $p'$ cannot occur as the marginals of a true-distribution $p$. However, perhaps surprisingly, this is never the case: any pseudo-distribution $p'$  can be (efficiently) implemented through a true distribution $p$:

\begin{lemma}\label{lemma:rounding} Let $\ell \geq 1$ and $p' \in \Delta_\ell([n])$ be arbitrary.
Then, there exists a distribution $p \in \Delta\left([n]^{(\ell)}\right)$
such that $\mathbb{P}_{S \sim p}(i \in S) = p'_i$ for $i \in [n]$. Moreover, such a $p$ exists with support of size at most $n$ and can be computed in $O(n^2)$ time given $p'$. If computing $p$ explicitly is not required, we can sample from such a $p$ 
in $O(n)$ time per sample.
\end{lemma}
\begin{proof} The linear time per sample part follows from applying dependent randomized rounding \cite{gandhi_dependent_rounding} to a star graph with $n$ edges of weights $p_1', \dots, p_n'$ (even ignoring the dependent part). The quadratic explicit construction follows by applying the \emph{AllocationFromShares} algorithm in \cite{harris_randomized_rounding}. See also \cite{grimmett_rounding} and the informative discussion in \cite{suzuki_randomized_rounding}.

We note that other existential proofs have appeared (sometimes implicitly) in the literature, using tools like the Birkhoff-von Neumann theorem and Carathéodory's theorem \cite{conitzer_lp_security,selection_minmax_regret_randomized}. Given existence, a distribution can be computed \cite{selection_minmax_regret_randomized} by setting up an LP with exponentially-many variables and polynomially-many constraints and solving its dual using the ellipsoid method with a poly-time separation oracle \cite{grotschel_ellipsoid, schrijver_ellipsoid} (however, not in strongly polynomial time).
\end{proof}

As a result, our new goal is to give a linear-time algorithm outputting a pseudo-distribution $p' \in \Delta_\ell([n])$ that maximizes $\val{p'}$. The rest of the paper is dedicated to this task. Knowing $p'$, we can then apply \cref{lemma:rounding} to get our main result: \cref{th:main}.

\subsection{Value-Maximizing Pseudo-Distributions}

In order not to overburden the notation, henceforth, we drop the apostrophe and ask for a pseudo-distribution $p \in \Delta_\ell([n])$ that maximizes $\val{p}$. It will also be convenient to, without loss of generality, allow the elements in $p$ to sum to \emph{at most} $\ell$ instead of exactly $\ell$: we write $p \in \Delta_{\leq \ell}([n])$ from now on to signal this fact.

To achieve our goal, we would now like to show similar structural results to \cref{lemma:increasing:ell1,lemma:exact-shape-ell-one}. For \cref{lemma:increasing:ell1}, a similar argument shows that it suffices to consider pseudo-distributions satisfying $v_1 \cdot p_1 \geq \dots \geq v_n \cdot p_n$. However, an immediate analog of \cref{lemma:exact-shape-ell-one} is not possible. Instead, by going through an intermediary lemma, we show that it suffices to restrict ourselves to what we call \emph{nice pseudo-distributions}, which admit a nice geometric interpretation through a water-in-vessels analogy. 
We achieve our final goal by a careful, efficient simulation of continuously decreasing the water level in the first vessel. We begin with the following lemma, whose proof is very similar to that of \cref{lemma:increasing:ell1} and hence deferred to the appendix.

\begin{restatable}{lemma}{lemmaincreasing}\label{lemma:increasing} There exists a pseudo-distribution $p \in \Delta_{\leq \ell}([n])$ maximizing $\val{p}$ that satisfies $v_1 \cdot p_1 \geq \dots \geq v_n \cdot p_n.$
\end{restatable}

\begin{lemma}\label{lemma:increasing-with-equalities} There exists a pseudo-distribution $p \in \Delta_{\leq \ell}([n])$ maximizing $\val{p}$ satisfying $v_1 \cdot p_1 = \dots = v_{t + 1} \cdot p_{t + 1} \geq \dots \geq v_n \cdot p_n.$
\end{lemma}
\begin{proof} 
By \cref{lemma:increasing}, let $p$ be a maximizer of $\val{p}$ such that $v_1 \cdot p_1 \geq \dots \geq v_n \cdot p_n.$ If $p$ satisfies the required property, then we are done. Otherwise, construct $p' : [n] \to [0, 1]$ as follows:
\begin{equation*}
    p'_j = \left\{
    \begin{array}{ll}
          \frac{v_{t + 1} \cdot p_{t + 1}}{v_j} & j \leq t \\
          p_j & j > t \\
    \end{array}
    \right.
\end{equation*}
Note that for $j \leq t$ we have $p'_j = \frac{v_{t + 1} \cdot p_{t + 1}}{v_j} \leq \frac{v_j \cdot p_j}{v_j} = p_j,$ so $\sum_{j = 1}^{n}p'_j \leq \sum_{j = 1}^{n}p_j \leq \ell$, from which $p' \in \Delta_{\leq \ell}([n])$. Moreover, by construction, we have $\val{p'} = \val{p}$ and that $p'$ satisfies $v_1 \cdot p_1 = \dots = v_{t + 1} \cdot p_{t + 1} \geq \dots \geq v_n \cdot p_n$.
\end{proof}

\begin{figure}[t]
    \centering
    \begin{subfigure}[b]{1.00\linewidth}
    \centering
    \includegraphics[scale=0.30]{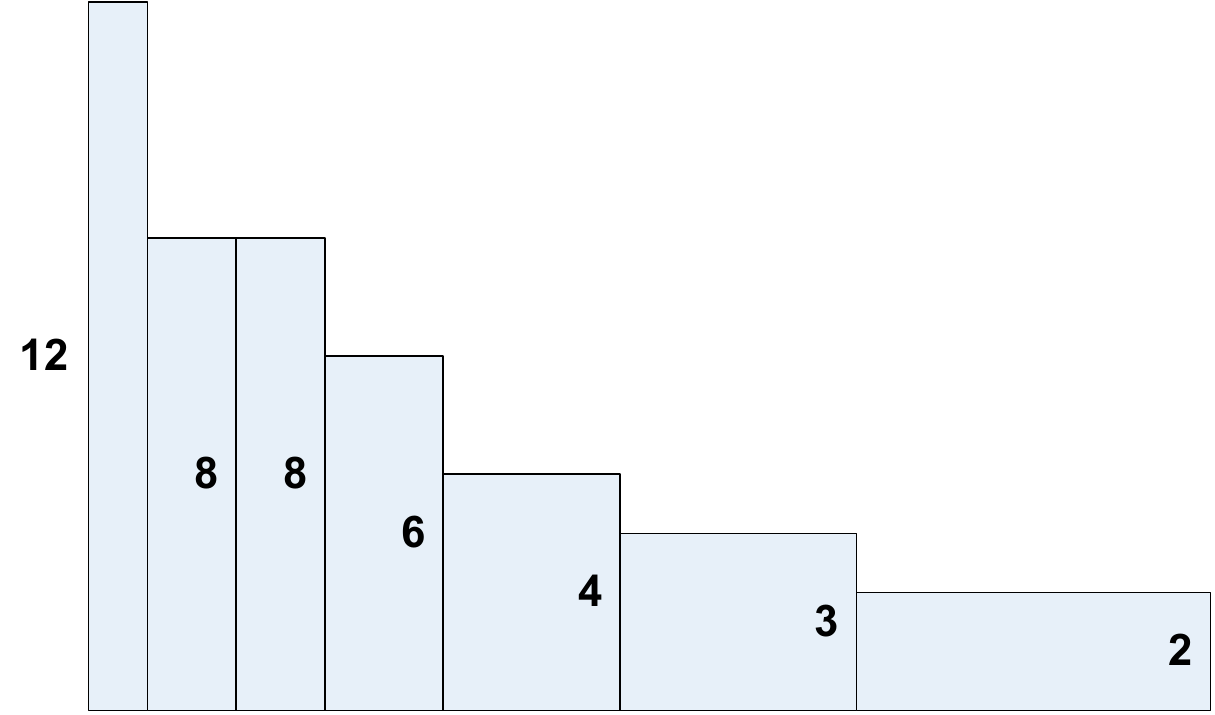}
    \caption{Example instance for $n = 7, t = 1, \ell = 5.$}
    \label{fig:water-level-1}
    \end{subfigure}
    
    \begin{subfigure}[b]{1.00\linewidth}
    \centering
    \hspace{0.34cm}
    \includegraphics[scale=0.30]{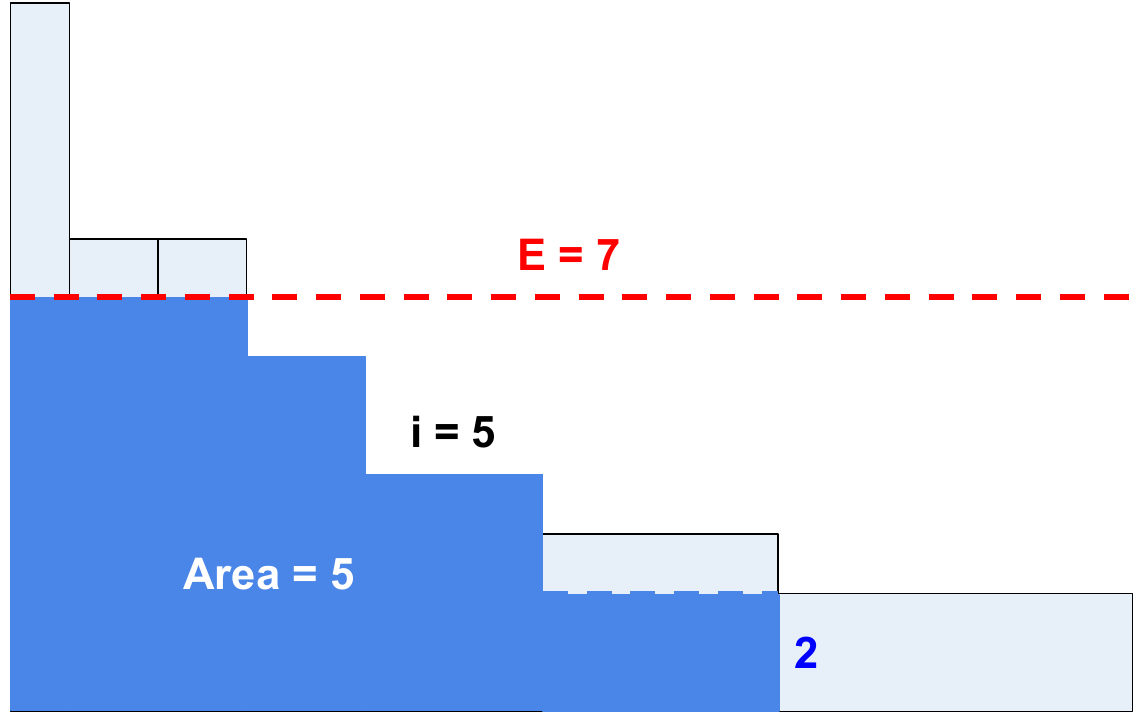}
    \caption{The maximal $7$-nice pseudo-distribution $p.$}
    \label{fig:water-level-2}    
    \end{subfigure}
    \caption{Consider an example with $n = 7, t = 1, \ell = 5$ and $\mathbf{v} = (12, 8, 8, 6, 4, 3, 2).$ This is depicted in \cref{fig:water-level-1} by rectangles with heights given by $\mathbf{v},$ each of area 1. One can understand pseudo-distributions $p$ using a water-filling metaphor:~each rectangle corresponds to a container comprising one unit of volume and each value $p_i \in [0, 1]$ corresponds to pouring $p_i$ units of water into the $i$-th container. By the choice of widths, pouring $p_i$ units of water into the $i$-th container makes the water rise to height $h_i := v_i \cdot p_i$ inside the container. Given $(t, \ell),$ a pseudo-distribution is $(E, i)$-nice if: (0) it uses at most $\ell$ units of water; (1) the water rises to level $E$ in the first $t + 1$ containers; (2) each container $k$ among the first $i$ is \emph{saturated} for $E$, i.e., either the water rises to height $E$ in $k$ or $E < v_k$ and $k$ is full; (3) container $i + 1$ (if it exists) is not saturated for $E$; (4) all subsequent containers are empty. Given $E$, there can exist at most one maximal $E$-nice pseudo-distribution: fill in the first $t + 1$ containers to level $E,$ if this exceeds the water budget $\ell,$ then no solution exists, otherwise continue in order through the next containers, saturating them until there is not enough water left to saturate current container. This is demonstrated for $E = 7$ in \cref{fig:water-level-2}: the first $3 \geq t + 1$ containers rise to level $7$, the next two containers have $p_4 = p_5 = 1$ (hence $i = 5$), the following container is not saturated: $p_6 = \frac{2}{3}$, and the last container is empty: $p_7 = 0.$ Overall, $\mathbf{p} = (\frac{7}{12}, \frac{7}{8}, \frac{7}{8}, 1, 1, \frac{2}{3}, 0),$ whose entries sum up to $\ell = 5.$}  
    \label{fig:water-level}
\end{figure}

\begin{definition} A pseudo-distribution $p \in \Delta_{\leq \ell}([n])$ is $(E, i)$-\emph{nice} for some $E \geq 0$ and $i \geq t + 1$ if the following conditions are satisfied:
\begin{enumerate}
    \item $v_k \cdot p_k = E$ for all $1 \leq k \leq t + 1$;
    \item $p_k = \min\{1, \frac{E}{v_k}\}$ for all $1 \leq k \leq i$;
    \item $p_{i + 1} < \min\{1, \frac{E}{v_{i + 1}}\}$;\footnote{Provided $i + 1 \leq n,$ otherwise we consider the last two conditions vacuously true.}
    \item $p_k = 0$ for all $i + 2 \leq k \leq n.$
\end{enumerate}
Moreover, $p$ is $E$-\emph{nice} if it is $(E, i)$-nice for some $i \geq t + 1$ and it is \emph{nice} if it is $E$-nice for some $E \geq 0$.

An $E$-nice pseudo-distribution $p \in \Delta_{\leq \ell}([n])$ is \emph{maximal} (for $E$) if there is no $E$-nice pseudo-distribution $p' \in \Delta_{\leq \ell}([n])$ such that $p'_i \geq p_i$ for all $1 \leq i \leq n$ and $p'_i > p_i$ for at least one $i.$
\end{definition}

The notion of nice pseudo-distributions might appear complex at first. However, it has an elegant geometric interpretation by a water-filling metaphor presented in \cref{fig:water-level} and its caption.

\begin{lemma}\label{lemma:full-characterization} There exists a pseudo-distribution $p \in \Delta_{\leq \ell}([n])$ maximizing $\val{p}$ that is nice.
\end{lemma}
\begin{proof}%
By \cref{lemma:increasing-with-equalities}, let $p$ be a maximizer of $\val{p}$ such that $v_1 \cdot p_1 = \dots = v_{t + 1} \cdot p_{t + 1} \geq \dots \geq v_n \cdot p_n$. Define $E := v_1 \cdot p_1$ and let $t + 1 \leq i \leq n$ be the largest index such that $p_i = \min\{1, \frac{E}{v_i}\}$. Among such pseudo-distributions $p$, further assume that we selected one maximizing the pair $(i, p_{i + 1})$ lexicographically, where we consider $p_{n + 1} := 0$ in order for this to always be well-defined. If $i = n,$ then $p$ is $(E, i)$-nice, so we are done, otherwise, assume $i + 1 \leq n,$ implying $p_{i + 1} < \min\{1, \frac{E}{v_{i + 1}}\}$, so $p$ satisfies the first three conditions for being $(E, i)$-nice. If the forth condition is also satisfied, then we are done. Otherwise, write $x := \sum_{k = i + 2}^n p_k,$ and observe that $x > 0$. Furthermore, note that $p_i \cdot v_i > p_{i + 1} \cdot v_{i + 1}.$ To see this, assume the contrary, i.e., $p_i \cdot v_i = p_{i + 1} \cdot v_{i + 1},$ from which $p_{i + 1} = \frac{v_i \cdot p_i}{v_{i + 1}} = \frac{v_i}{v_{i + 1}} \min\{1, \frac{E}{v_i}\} = \min\{\frac{v_i}{v_{i + 1}}, \frac{E}{v_{i + 1}}\} \geq \min\{1, \frac{E}{v_{i + 1}}\},$ a contradiction. Now, define $\varepsilon := \min\{x, \frac{v_i \cdot p_i}{v_{i + 1}} - p_{i + 1}\}$ and note that $\varepsilon > 0.$\footnote{This is because $x > 0$ and $p_i \cdot v_i > p_{i + 1} \cdot v_{i + 1}.$} Construct $p' : [n] \to [0, 1]$ as follows:%
\begin{equation*}
    p'_k = \left\{
    \begin{array}{ll}
          p_k & k \leq i \\
          p_k + \varepsilon & k = i + 1 \\
          p_k \cdot \frac{x - \varepsilon}{x} & k \geq i + 2
    \end{array}
    \right.
\end{equation*}

Let us first show that $p' \in \Delta_{\leq \ell}([n])$. For this, we need to show that all entries are non-negative and sum up to at most $\ell$. Non-negativity follows since $\frac{x - \varepsilon}{x} \geq 0,$ and the sum is a simple algebraic calculation.

Let us now show that $p'$ satisfies $v_1 \cdot p'_1 = \dots = v_{t + 1} \cdot p'_{t + 1} \geq \dots \geq v'_n \cdot p'_n.$ The equalities hold by construction of $p'$ from $p$ and because $i \geq t + 1.$ It then remains to show that for all $1 \leq k \leq n - 1$ we have $v_k \cdot p'_k \geq v_{k + 1} \cdot p'_{k + 1}.$ The cases $k < i$ and $k \geq i + 2$ are immediate by construction. The case $k = i + 1$ is also straightforward: $v_{i + 1} \cdot p'_{i + 1} \geq v_{i + 1} \cdot p_{i + 1} \geq v_{i + 2} \cdot p_{i + 2} \geq v_{i + 2} \cdot p'_{i + 2}.$ The final case $i = k$ amounts to $v_i \cdot p_i \geq v_{i + 1} \cdot (p_{i + 1} + \varepsilon),$ which holds by $\varepsilon$'s definition.

Finally, let us show that the existence of $p'$ is a contradiction. First, $\val{p'} \geq \val{p}.$ This is because $p'$ is constructed from $p$ by moving $\varepsilon$ probability mass from positions $i + 2, \dots, n$ to position $i + 1,$ and $v_{i + 1}$ is no smaller than $v_{i + 2}, \dots, v_n$ and hence $\val{p'} = \sum_{k = t + 1}^{n} v_k \cdot p'_k \geq \sum_{k = t + 1}^{n} v_k \cdot p_k = \val{p}.$ If $\val{p'} > \val{p}$ then this is already a contradiction, otherwise, $\val{p'} = \val{p}$ holds, but even in that case $p'$ has by construction a lexicographically strictly higher $(i, p'_{i + 1})$ pair, again a contradiction.
\end{proof}

Before proceeding further, we emphasize the observation made in the caption of \cref{fig:water-level}: Given $E$, there exists at most one maximal $E$-nice pseudo-distribution, obtained by the presented left-to-right water-filling argument. In fact, something even stronger holds: $E$-nice pseudo-distributions are linearly ordered by the amount of water used: the only way to create non-maximal $E$-nice pseudo-distributions is to follow the same water-filling argument but stop it early before it had a chance to use all possible water. For a fixed $E$, pseudo-distributions with more water can not lead to a worse value, so we can augment \cref{lemma:full-characterization} to get that it suffices to look at \emph{maximal} pseudo-distributions, which are uniquely determined by $E$ whenever they exist. Let us give the existence conditions in light of the water-filling argument: (i) $E \leq v_{t + 1},$ as otherwise $p_k = \frac{E}{v_k} > 1$ would hold for some $1 \leq k \leq t + 1,$ in particular for $k = t + 1$
and (ii) $\sum_{k = 1}^{t + 1}\frac{E}{v_k} \leq \ell \iff E \leq \ell \cdot (\sum_{k = 1}^{t + 1}\frac{1}{v_k})^{-1},$ since otherwise the $\ell$ units of water we have at our disposal are insufficient to make the water rise to level $E$ in the first $t + 1$ containers. Hence, (maximal) $E$-nice pseudo-distributions exist for $0 \leq E \leq E_\mathit{max} := \min\{v_{t + 1}, \ell \cdot (\sum_{k = 1}^{t + 1}\frac{1}{v_k})^{-1}\}$. It remains to identify $E \in [0, E_\mathit{max}]$ yielding a maximum-value $E$-nice pseudo-distribution. The main idea to do this in linear time is to note that the value of the maximal $E$-nice distribution is piece-wise linear in $E$. We prove this fact and give a linear-time algorithm that produces and iterates through the breakpoints in order from $E_\mathit{max}$ down to $0$, at the same time computing the value of the corresponding pseudo-distributions and outputting the best one at the end (the maximum has to happen at a break-point due to piece-wise linearity). The details can be found in the proof of the lemma below, which is rather involved and hence deferred to the appendix. To give a different view, our algorithm simulates the continuous process of decreasing the water level in the first container (i.e., $E$, starting at $E_\mathit{max}$), outputting the best achievable value along the way in overall linear time. 

\begin{restatable}{lemma}{continoussimulation}\label{lemma:continous-simulation} A pseudo-distribution $p \in \Delta_{\leq \ell}([n])$ maximizing the quantity $\val{p}$ can be computed in time $O(n)$.
\end{restatable}

\section{Conclusion and Future Work}

We introduced the Byzantine Selection Problem and gave attractive linear-time algorithms sampling from a value-maximizing distribution. It would also be interesting to study the problem for other notions of optimality inspired by algorithm design under uncertainty, such as minimizing regret or maximizing the competitive ratio. Introducing an online element, where boxes arrive one at a time, like in the secretary problem, could also prove fruitful. Finally, our model followed the distributed computing literature and assumed a known threshold $t$ on the number of empty boxes (we assumed exactly $t$ purely for analysis purposes). In practice, such a threshold may not be available, in which case one can still use our algorithms with $t = n - 1$ (which might, however, be overly conservative). One can choose to work with a lower $t$ depending on their risk aversion. Studying how the optimum varies with $t$ for a fixed instance or assuming a prior distribution on $t$ could be promising avenues for future work.

More broadly, our paper argues for an exciting yet relatively unexplored research avenue: introducing byzantine elements into classic non-distributed settings, such as the centralized settings emerging in social choice theory. The recent work in \cite{constantinescu2025byzantinestablematching} brings this element to stable matchings. While their paper targets a distributed setting, one of their contributions is the development of a fault-tolerant notion of stable matchings. It would be appealing to conduct similar investigations in other application domains, such as preference aggregation or fair division. The paper \cite{eliaz} targets so-called \emph{fault-tolerant implementations}, and could also serve as a good starting point for further investigation in this direction when incentives are to be considered as well.

\begin{acks}
We thank Judith Beestermöller, Joël Mathys, Jannik Peters, Jiarui Gan, and Edith Elkind for the many productive discussions that helped shape the paper in its current form and for highlighting relevant literature. We are, moreover, grateful to the anonymous reviewers for their insightful comments and suggestions that helped improve the quality of the paper.
\end{acks}

\bibliographystyle{ACM-Reference-Format} 
\bibliography{references}


\newpage
\ 
\newpage 

\appendix

\section{Omitted Proofs}

This appendix contains the proofs omitted from the main body, namely those of \cref{lemma:increasing,lemma:continous-simulation}.

\lemmaincreasing*
\begin{proof} The proof proceeds analogously to the proof of \cref{lemma:increasing:ell1}. Instead of assuming at the beginning that $p$ maximizes $\val{p}$ with the no-loss-of-generality assumption that for any $i < j$ with $v_i = v_j$ we have $p_i \geq p_j$, we will now additionally assume that among such maximizing pseudo-distributions, $p$ has the smallest possible $\sum_{k = 1}^{n}p_k \leq \ell.$ We then proceed as before, defining $p'.$ The calculation that previously showed $\sum_{k = 1}^{n}p'_k < 1$ will now instead show that $\sum_{k = 1}^{n}p'_k < \sum_{k = 1}^{n}p_k,$ contradicting minimality of the sum. (No need to define $p''$ anymore.)
\end{proof}

\continoussimulation*
\begin{proof} Consider varying $E$ continuously starting from $E_\mathit{max} := \min\{v_{t + 1}, \ell \cdot (\sum_{k = 1}^{t + 1}\frac{1}{v_k})^{-1}\}$ down to $0.$ For a fixed $E,$ let $p^E$ be the corresponding maximal $E$-nice pseudo-distribution (which is well-defined as $E \in [0, E_\mathit{max}]$) 
and $i_E$ be the corresponding index such that $p^E$ is $(E, i_E)$-nice.  Moreover, let $k_E \geq t + 1$ be such that $v_{k_E} \geq E > v_{k_E + 1}$, assuming the convention that $v_{n + 1} := -1$ for uniformity of this definition. 
As $E$ decreases, $k_E$ monotonically increases. Moreover, as $E$ decreases, more water is pushed to the right, making $i_E$ also monotonically increase. As a result, the mapping $E \mapsto (i_E, k_E)$ can be described using a sequence $E_\mathit{max} = x_1 \geq \dots \geq x_m = 0$ with $m \leq 2n$ such that on each interval $(x_{q + 1}, x_q]$ for $1 \leq q < m$ the mapping is constant.

We are interested in finding an $E \in [0, E_\mathit{max}]$ that maximizes $\val{p^E}$, as from it the pseudo-distribution $p^E$ can be easily computed in linear time using the water-filling argument. Note that $\val{p^E}$ is continuous. The main idea of the proof will be to show that $\val{p^E}$ is additionally piece-wise linear with breakpoints (i.e., points where the slope can change) at $x_1, \dots, x_m$. Since a piece-wise linear function on a closed domain attains its maximum at one of the breakpoints, it will be enough to give a linear-time algorithm computing the breakpoints and evaluating $\val{p^E}$ at them. Our algorithm will proceed precisely along these lines. However, to remove the need to consider a number of edge cases, we will refine the approach slightly. In particular, it will follow that if for some $E$ we have $i_E = n$, then we can disregard all $E' < E$, meaning that we can ignore breakpoints after the first one with $i_E = n$. Our algorithm can be modified to compute all breakpoints and associated values, but this is not needed to get the maximum.

First, let us show that, indeed,
$\val{p^E}$ is piece-wise linear with breakpoints at $x_1, \dots, x_m$. Consider some $1 \leq q < m$ and any  $E \in (x_{q + 1}, x_q]$ and $\varepsilon > 0$ such that $E - \varepsilon \in [x_{q + 1}, x_q]$. We will show that $\val{p^{E - \varepsilon}} - \val{p^E} = \varepsilon \cdot C$, where $C$ does not depend on $\varepsilon$. Write $i := i_E = i_{E - \varepsilon}$ and $k := k_E = k_{E - \varepsilon}$.

To determine $\val{p^{E - \varepsilon}} - \val{p^E}$, let us compare $p^E$ and $p^{E - \varepsilon}$, i.e., start with $p = p^E$ and decrease $E$ by $\varepsilon$ to reach $p = p^{E - \varepsilon}$: how does the water move? The first $k$ containers lose water so as to lower the water level from $E$ to $E - \varepsilon$. Consider one of those containers $1 \leq j \leq k$. To be at water level $E$, container $j$ had $\frac{E}{v_j}$ units of water inside, and to be at water level $E - \varepsilon$ after $E$ is decreased, $\frac{E - \varepsilon}{v_j}$ units of water need to be left inside. Hence, in the process, it loses $\frac{\varepsilon}{v_j}$ units of water. Summing up over the $k$ containers, in total $\varepsilon \cdot \sum_{j = 1}^{k}\frac{1}{v_j}$ units of water are lost. Where does this water go? In container $i + 1$ if $i < n$, and nowhere otherwise, in which case it is lost. How does this movement of water impact $\val{p}$? Recall that $\val{p}$ is the sum of the water levels in containers $t + 1, \dots, n$. Hence, since the first $k \geq t + 1$ containers' water levels decreased by $\varepsilon$ and we are only interested in containers $t + 1, \dots, n$, we get that $\val{p}$ decreased by $(k - t) \cdot \varepsilon$. If $i < n$, i.e., the lost water went somewhere, we need to also account for the increase in water level of container $i + 1$. Since we know that in this case the amount of water added to this container is $\varepsilon\cdot\sum_{j = 1}^{k}\frac{1}{v_j}$, the water level in this container increased by $v_{i + 1} \cdot \varepsilon \cdot \sum_{j = 1}^{k}\frac{1}{v_j}$. Summarizing:%
\begin{equation}\label{eq:diff-value}
    \val{p^{E - \varepsilon}} - \val{p^E} = \left\{
    \begin{array}{ll}
           \varepsilon \cdot \sum_{j = 1}^{k}\frac{v_{i + 1}}{v_j} - (k - t) \cdot \varepsilon, & i < n \\
           - (k - t) \cdot \varepsilon, & i = n 
    \end{array}
    \right.
\end{equation}

Observe that, in both cases, $\varepsilon$ can be factored out, so we have shown that, indeed, $\val{p^{E - \varepsilon}} - \val{p^E} = \varepsilon \cdot C$, where $C$ does not depend on $\varepsilon$, as desired. 

Note that \cref{eq:diff-value} is useful more generally than in proving piece-wise linearity: if we set $E = x_j$ for some $1 \leq j < m$ and $\varepsilon = x_j - x_{j + 1}$, it gives a formula for $\val{p^{x_{j + 1}}} - \val{p^{x_j}}$ whose value is non-positive when $i_{x_j} = n$. Hence, let $1 \leq m' \leq m$ be the first index such that $i_{m'} = n$ and note that, by the previous, the maximum $\val{p^E}$ is attained at one of the points $x_1, \dots, x_{m'}$. 

We will now give a linear-time algorithm that computes the breakpoints $x_1, \dots, x_{m'}$ in
order and finds the associated $\val{p^E}$ at them. The value $m'$ is not known a priori, but the algorithm will detect when it has been reached and terminate. We will assume a precomputed prefix-sums array for $\frac{1}{v_1}, \dots, \frac{1}{v_n}$, so that sums of the kind present in \cref{eq:diff-value} can be determined in constant time. The algorithm maintains there variables $(E, X, k)$, which start initialized to $(x_1, \val{p^{x_1}}, k_{x_1})$, where recall that $x_1 = E_\mathit{max}$. This initialization is possible in linear time using the water-filling argument.
At the beginning of each step $1 \leq j \leq m'$ of the algorithm, $(E, X, k) = (x_j, \val{p^{x_j}}, k_{x_j})$ will hold. The step will either determine that $j = m'$ and terminate, or compute the required values and set $(E, X, k) \gets (x_{j + 1}, \val{p^{x_{j + 1}}}, k_{x_{j + 1}})$, in preparation for the next step.

First, let us see how to check whether $j = m'$ holds. This amounts to checking whether $i_E = n$, for which we need to compute $i_E$ (which we will also need for other purposes).
This can be achieved using the formula $i_E = \min\{n, k + \floor{\ell - \sum_{j = 1}^k\frac{E}{v_j}}\}$. The idea behind is as follows: out of the total of $\ell$ units of water, $\sum_{j = 1}^k\frac{E}{v_j}$ are used up by the first $k$ containers, and the remaining $\ell - \sum_{j = 1}^k\frac{E}{v_j}$ are left to fill containers $k + 1, \dots, n$ in this order, filling up one container completely before going on to the next one. There is enough water to completely fill $\floor{\ell - \sum_{j = 1}^k\frac{E}{v_j}}$ additional containers, from which the conclusion follows.

Now, assume $j < m'$ (i.e., $i_E < n$, from which also $k = k_E \leq i_E < n$ holds), so we want to compute $(x_{j + 1}, \val{p^{x_{j + 1}}}, k_{x_{j + 1}})$ for the next iteration. The most interesting part will be computing $x_{j + 1}$ (the other two will follow easily afterward). To compute $x_{j + 1}$, consider starting with $E$ and letting it decrease continuously so that its value throughout the process is $E - \varepsilon$ with $\varepsilon$ steadily increasing. We would like to stop the process when $E - \varepsilon = x_{j + 1}$ holds. How can we quickly determine the value of $\varepsilon$ for which this happens, so that we do not have to simulate the continuous process? During the process, two things can happen, and the one happening first is what stops the process: either $k_{E - \varepsilon} \neq k_{E}$, which first happens when $E - \varepsilon = v_{k + 1}$, or $i_{E - \varepsilon} \neq i_{E}$, which first happens when $p^{E - \varepsilon}_{i_E + 1} = 1$, i.e., container $i_E + 1$ gets filled. If the latter were to happen before the former, then the latter occurs when $p^E_{i_E + 1} + \varepsilon \cdot \sum_{j = 1}^k \frac{1}{v_j} = 1$. By the same reasoning as the one we used to find $i_E$, the value $p^E_{i_E + 1}$ can be determined as the fractional part of $\ell - \sum_{j = 1}^k\frac{E}{v_j}$, denoted by $\{\ell - \sum_{j = 1}^k\frac{E}{v_j}\}$. Hence, the process stops when $\varepsilon = \varepsilon_\mathit{min} := \min\{E - v_{k + 1}, \frac{1 - p^E_{i_E + 1}}{\sum_{j = 1}^k \frac{1}{v_j}}\}$, depending on which of the two cases occurs first. As a result, we can compute $x_{j + 1}$ as $x_j - \varepsilon_\mathit{min}$.

It remains to show how to compute $ \val{p^{x_{j + 1}}}$ and $k_{x_{j + 1}}$. To get $\val{p^{x_{j + 1}}}$, we use the fact that we know $X = \val{p^{x_j}}$ and employ \cref{eq:diff-value} to determine $\val{p^{x_{j + 1}}} - \val{p^{x_j}}$, adding it to $X$ to get $\val{p^{x_{j + 1}}}$. To get $k_{x_{j + 1}}$, we employ monotonicity: we increase $k$ repeatedly by 1 for as long as $x_{j + 1} \leq v_{k + 1}$. This takes amortized constant time per step.

For completeness, the full algorithm is given below:

\begin{algorithmic}[1]
\State $E \gets E_\mathit{max}$.
\State Compute $p^E$ using the water-filling argument and determine $\val{p^E}$ and $k_E$.
\State $(X, k) \gets (\val{p^E}, k_E)$
\For{$j \gets 1, 2, \dots$} \Comment{Invar.: $(E, X, k) = (x_j, \val{p^{x_j}}, k_{x_j})$}
    \State Compute $i_E = \min\{n, k + \floor{\ell - \sum_{j = 1}^k\frac{E}{v_j}}\}$.
    \State \textbf{if} $i_E = n$ \textbf{then} \textbf{break} \Comment{Stop when $j = m'$.}
    \State Compute $p^E_{i_E + 1} = \{\ell - \sum_{j = 1}^k\frac{E}{v_j}\}$ \Comment{Fractional part.}
    \State $\varepsilon_\mathit{min} \gets \min\{E - v_{k + 1}, \frac{1 - p^E_{i_E + 1}}{\sum_{j = 1}^k \frac{1}{v_j}}\}$
    \State $X \gets X + \left(v_{i_E + 1} \cdot \varepsilon_\mathit{min} \cdot \sum_{j = 1}^{k}\frac{1}{v_j}\right) - (k - t) \cdot \varepsilon_\mathit{min}$ 
    \State $E \gets E - \varepsilon_\mathit{min}$
    \While{$k + 1 < n$ and $E \leq v_{k + 1}$}
        \State $k \gets k + 1$
    \EndWhile
\EndFor \qedhere
\end{algorithmic}
\end{proof}

\end{document}